\documentclass[copyright,creativecommons]{eptcs}
\providecommand{\event}{FICS 2024} 

\usepackage{amsmath,amsthm,amssymb,latexsym}
\usepackage{graphicx}
\usepackage{xcolor} 
\usepackage{xspace} 
\usepackage{xargs} 
	\usepackage{ebproof} 
\usepackage{mathtools}

\theoremstyle{plain}
\newtheorem{theorem}{Theorem}
\newtheorem{lemma}[theorem]{Lemma}

\theoremstyle{definition}
\newtheorem{definition}[theorem]{Definition}
\newtheorem{example}[theorem]{Example}

\usepackage{iftex}
\ifpdf
  \usepackage{underscore}         
  \usepackage[T1]{fontenc}        
\else
  \usepackage{breakurl}           
\fi

\newcommand{\calS}{\mathcal{S}}
\newcommand{\calT}{\mathcal{T}}

\newcommand{\bbS}{\mathbb{S}}

\newcommand{\sfC}{\mathsf{C}}

\newcommand{\sfR}{\mathsf{R}}
\newcommand{\sfS}{\mathsf{S}}

\renewcommand{\|}{~|~}

\renewcommand{\phi}{\varphi}

\newcommand{\powset}{\mathcal{P}}

\renewcommand{\land}{\wedge}
\renewcommand{\lor}{\vee}

\newcommand{\seq}{\Rightarrow}

\newcommand{\Clos}{\mathsf{Clos}}
\newcommand{\tracestep}{\to_{C}}

\newcommand{\Ax}{\ensuremath{\mathsf{Ax}}\xspace}
\newcommand{\RuAndL}{\ensuremath{\land_L}\xspace}
\newcommand{\RuAndR}{\ensuremath{\land_R}\xspace}
\newcommand{\RuNegL}{\ensuremath{\neg_L}\xspace}
\newcommand{\RuNegR}{\ensuremath{\neg_R}\xspace}

\newcommand{\RuDia}{\ensuremath{\mathsf{{\ldia}}}\xspace}

\newcommand{\RuFL}{\ensuremath{\lF_L}\xspace}
\newcommand{\RuFR}{\ensuremath{\lF_R}\xspace}
\newcommand{\RuWeak}{\ensuremath{\mathsf{w}}\xspace}
\newcommand{\RuWeakL}{\ensuremath{\mathsf{w}_L}\xspace}
\newcommand{\RuWeakR}{\ensuremath{\mathsf{w}_R}\xspace}
\newcommand{\RuCut}{\ensuremath{\mathsf{cut}}\xspace}

\newcommand{\RuU}{\ensuremath{\mathsf{u}}\xspace}
\newcommand{\RuF}{\ensuremath{\mathsf{f}}\xspace}

\newcommand{\Ru}{\ensuremath{\mathsf{R}}\xspace}
\newcommand{\RuDischarge}[1][\dx]{\ensuremath{\mathsf{D}^{#1}}\xspace}

\newcommand{\Tokens}{\mathcal{D}}
\newcommand{\dx}{\ensuremath{\mathsf{x}}}
\newcommand{\dy}{\ensuremath{\mathsf{y}}}
\newcommand{\dz}{\ensuremath{\mathsf{z}}}

\newcommand{\Prop}{\ensuremath{\mathsf{Prop}}\xspace}

\newcommand{\GKe}{\ensuremath{\mathsf{GKe}}\xspace}
\newcommand{\GKea}{\ensuremath{\mathsf{GKe}}\xspace}

\newcommand{\MLe}{\ensuremath{\mathsf{MLe}}\xspace}

\newcommand{\sccr}{\twoheadrightarrow}
\newcommand{\cyclicPT}{\calT_\pi^C}

\newcommand{\depth}{\ensuremath{\mathsf{depth}}\xspace}
\newcommand{\rank}{\ensuremath{\mathsf{rank}}\xspace}

\newcommand{\comp}{\ensuremath{\mathsf{comp}}\xspace}

\newcommand{\unf}[2][]{#2^{*_{#1}}}
\newcommand{\merge}[3]{\ensuremath{[#1]#2[#3]}\xspace}

\newcommand{\ldia}{ \scalebox{0.8}{\rotatebox[origin=c]{45}{$\square$}}}
\newcommand{\lbox}{\scalebox{0.8}{$\square$}}

\newcommand{\lm}{\mbox{\ooalign{\ld \cr \hidewidth\raise.05ex\hbox{$* \mkern3.1mu$}\cr}}} 
\newcommand{\lbm}{\ooalign{$\lbox$ \cr \hidewidth\raise.05ex\hbox{$* \mkern5.5mu$}\cr}\hspace{-0.1cm}} \newcommand{\lcm}{\ooalign{$\lbox$ \cr \hidewidth\raise.05ex\hbox{$\cdot \mkern2.9mu$}\cr}}

\newcommand{\lF}{\ensuremath{\mathsf{F}}\xspace}

\newcommand{\anColor}[1]{\textcolor{black}{#1}}

\title{Cut elimination for Cyclic Proofs: \\ A Case Study in Temporal Logic\thanks{This work was supported by the Dutch Research Council [OCENW.M20.048, 617.001.857] and Knut and Alice Wallenberg Foundation [2020.0199].}}
\author{Bahareh Afshari
\institute{University of Gothenburg\\ Gothenburg, Sweden}
\email{bahareh.afshari@gu.se}
\and
Johannes Kloibhofer
\institute{University of Amsterdam\\
Amsterdam, The Netherlands}
\email{j.kloibhofer@uva.nl}
}
\def\titlerunning{Cut-elimination for Cyclic Proofs}
\def\authorrunning{B.~Afshari \& J.~Kloibhofer}
\begin{document}
\maketitle
\begin{abstract}
We consider modal logic extended with the well-known temporal operator `eventually'
and provide a cut-elimination procedure for a cyclic sequent calculus that captures this fragment. 
The work showcases an adaptation of the reductive cut-elimination method to cyclic calculi. Notably, the proposed algorithm applies to a cyclic proof and directly outputs a cyclic cut-free proof without appealing to intermediate machinery for regularising the end proof.
\end{abstract}
\section{Introduction}
%
\emph{Non-wellfounded} derivation systems underwent their first rigorous study in~\cite{Niwinski1996,Santocanale2002,Brotherston2006, schutte77}. In contrast to the traditional finite derivation trees, one allows for infinite branches that mimic the induction axioms required for capturing semantics of temporal operators (more generally, implicit/explicit fixpoints).
\textit{Cyclic proofs} constitute the regular fragment of non-wellfounded proofs and are closer in spirit, and often in practice, to finitary proof systems with induction axioms.
To exclude vicious reasoning and capture the true semantics of temporal constructs, a so called \emph{global validity condition} is placed on infinite branches. This condition necessitates a re-investigation of reductive cut elimination in the non-wellfounded setting.

With the advance of non-wellfounded and cyclic proof theory, cut elimination has also received considerable attention (see e.g.~\cite{Fortier2013, Das2018,Savateev2020,Saurin23}). The common approach proceeds in two steps: first cuts are pushed `up' away from the root  by the usual cut reductions to obtain, in the limit, a non-wellfounded derivation; in the second step the acquired limit structure is shown to satisfy the global validity condition.
When it comes to cyclic systems, cut-elimination procedures that can \textit{directly} produce cut-free {cyclic} proofs are rare. Although the aforementioned two-step approach can be applied, the resulting structure may not necessarily be a regular tree. For those calculi whose cyclic fragment does exhaust all validities one may invoke other machinery, such as automata, to find a cut-free cyclic proof.

We consider a simple temporal logic \MLe with a complete cyclic sequent calculus \GKe and describe a syntactic cut-elimination procedure for \GKe that works directly on cyclic proofs. The logic \MLe extends modal logic with the so called `eventually' operator, the reflexive and transitive closure of the diamond modality.  
Although $\MLe$ is expressively weak, this work can be seen as a starting point to apply cut elimination to cyclic proof systems for richer modal fixpoint logics.

Cuts occurring in a cyclic proof can be split into two categories: cuts that reside in a cycle and those that do not. We call them, respectively, \emph{unimportant} and \emph{important} cuts. We show that unimportant cuts are for most part harmless, in the sense that the cut formula does not interfere with the global validity condition and as such, these cuts can be eliminated as usual by pushing them away from the root of the proof-tree. 
The main challenge posed is to eliminate important cuts, as in this case the cut formula may contribute to the global validity condition.

The global validity condition can take different forms. A \emph{trace-based condition} is formulated in terms of traces, i.e. paths of ancestry of formulas, on a branch. Such conditions are relatively easy to formulate but harder to work with.
By annotating sequents with additional information one may succeed in formulating a \emph{branch-based condition} that  is not directly concerned with traces. Such a condition is easier to work with when performing operations such as projecting a proof or zipping two proofs that are critical in reductive cut elimination. 
Annotated proof systems for modal $\mu$-calculus \cite{Bradfield2007} were first introduced by Jungteerapanich and Stirling \cite{Jungteerapanich2010, Stirling2014} building on earlier work by Niwi\'nski and  Walukiewicz \cite{Niwinski1996}. 
There are general ways of obtaining annotated proof systems from cyclic proof systems \cite{Dekker2023, 
LW23}. It is with this rich framework in mind that the presented case analysis has been carried out.
%
\section{Modal Logic with Eventuality}
%
Let \Prop be an infinite set of propositions. \emph{Formulas} of \MLe are defined inductively by	
	\begin{align*}
		\phi := p \| \neg \phi \| \phi \land \phi \|  \ldia  \phi \| \lF \phi
	\end{align*}	
	where $p \in \Prop$.	
The connectives $\lor, \rightarrow$ and $\lbox$ are taken as abbreviations and can be defined as usual. 
To see \MLe as a fragment of the modal $\mu$-calculus, simply define $\lF \phi \equiv \mu x. (\phi \lor \ldia x)$.

Semantics of \MLe-formulas is given with respect to Kripke models in the standard way.
	A \emph{Kripke model} is a tuple $\bbS = (S,R,V)$, where $S$ is a non-empty set, $R$ is a binary relation on $S$ and $V$ is a function $S \rightarrow \powset(\Prop)$. 
	Let $s \in S$ and define $R^*$ to be the reflexive and transitive closure of $R$. The relation $\Vdash$ is defined inductively by:
	\begin{align*}
		\bbS,s &\Vdash p  				&&\Leftrightarrow 	&&p \in V(s) \\
		\bbS,s &\Vdash \neg \phi  		&&\Leftrightarrow  	&&\bbS,s \not\Vdash \phi \\
		\bbS,s &\Vdash \phi \land \psi	&&\Leftrightarrow 	&&\bbS,s \Vdash \phi \text{ and } \bbS,s \Vdash\psi\\
		\bbS,s &\Vdash \ldia \phi		&&\Leftrightarrow  	&&\text{there exists } t \in S \text{ with } sRt \text{ and } \bbS,t \Vdash \phi \\
		\bbS,s &\Vdash \lF\phi 			&&\Leftrightarrow 	&&\text{there exists } t \in S \text{ with } sR^*t \text{ and } \bbS,t \Vdash \phi
	\end{align*}
A formula $\phi$ is called \emph{valid} if $\bbS,s \Vdash \phi$ for every Kripke model $\bbS= (S,R,V)$ and every $s \in S$.

We present the cyclic annotated proof system \GKe, which is sound and complete for \MLe. The annotations allow us to formulate the global validity condition as a simple condition on paths.
Since \MLe is a fragment of the alternation-free $\mu$-calculus, as established in \cite{Marti2021}, one annotation mark per formula suffices for the formulation of the validity condition. Moreover, \MLe-traces are very well-behaved; in particular, cyclic traces do not pass through disjunctions, requiring at most one annotated formula in each sequent, as demonstrated in~\cite{Rooduijn2021}.

An \emph{annotated formula} is a pair $(\phi,a)$, usually denoted as $\phi^a$, where $a \in \{f,u\}$. We call  annotated formulas of the form $\phi^f$ \emph{in focus} and $\phi^u$ \emph{out of focus}. Let $\Gamma$ and $\Delta$ be \emph{sets} of annotated formulas, where every formula in $\Delta$ is out of focus and at most one formula in $\Gamma$ is in focus and this formula has the form $\lF \phi$ or $\ldia \lF \phi$ for some $\phi$. Then $\Gamma \seq \Delta$ is called an \emph{annotated sequent}. As every formula in $\Delta$ is out of focus, we omit the annotations in the right side of the annotated sequent. 
We define $\ldia \Delta = \{\ldia \phi \| \phi \in \Delta\}$ and, for a set of annotated formulas $\Gamma$, we set $\Gamma^u = \{\phi^u \| \phi^a \in \Gamma\}$. 
If it is clear from the context, we will  call annotated formulas/sequents simply formulas/sequents. 

The rules of \GKe for the propositional connectives and the modal operator are standard. The rules \RuFL and \RuFR for the eventually operator $\lF$ stem from the identity $\lF \phi \equiv \phi \lor \ldia \lF \phi$. In \RuFL the formula $\lF \phi^a$ is called \emph{principal}. 
The rule \RuFL and the focus rules \RuU and \RuF are the only ones that change the annotations of formulas. 
The \emph{discharge rule} \RuDischarge[] allows us to discharge leaves, that are labelled by sequents already occurring at one of its ancestors. To qualify as a \GKe proof, this is only allowed if a certain success-condition is met. Each instance of \RuDischarge[\dx] is labelled by a unique \emph{discharge token} $\dx$ taken from a fixed infinite set $\Tokens = \{\dx,\dy,\dz,...\}$.

\begin{figure}[t]
	\begin{minipage}{\textwidth}
		\begin{minipage}{0.20\textwidth}
			\begin{prooftree}
				\infer[left label=\Ax:]0{ p^{\anColor{u}} \seq p}
			\end{prooftree}
		\end{minipage}
		\begin{minipage}{0.22\textwidth}
			\begin{prooftree}
				\hypo{\phi^{\anColor{a}} \seq \Delta}
				\infer[left label= \RuDia:]1{ \ldia \phi^{\anColor{a}} \seq \ldia \Delta}
			\end{prooftree}
		\end{minipage}
		\begin{minipage}{0.250\textwidth}
			\begin{prooftree}
				\hypo{\phi^{\anColor{u}},\psi^{\anColor{u}}, \Gamma \seq \Delta}
				\infer[left label= \RuAndL:]1{\phi \land \psi^{\anColor{u}}, \Gamma \seq \Delta}
			\end{prooftree}
		\end{minipage}
		\begin{minipage}{0.3\textwidth}
			\begin{prooftree}
				\hypo{\Gamma \seq \Delta,\phi}
				\hypo{\Gamma \seq \Delta,\psi}
				\infer[left label= \RuAndR:]2{\Gamma \seq \Delta, \phi \land \psi}
			\end{prooftree}
		\end{minipage}
	\end{minipage}
	
	\bigskip
	
	\begin{minipage}{\textwidth}
		\begin{minipage}{0.35\textwidth}
			\begin{prooftree}
				\hypo{\ldia \lF \phi^{\anColor{a}},\Gamma \seq \Delta}
				\hypo{\phi^{\anColor{u}},\Gamma \seq \Delta}
				\infer[left label= \RuFL:]2{ \lF \phi^{\anColor{a}},\Gamma \seq  \Delta}
			\end{prooftree}
		\end{minipage}
		\begin{minipage}{0.24\textwidth}
			\begin{prooftree}
				\hypo{\Gamma \seq \Delta, \phi, \ldia \lF \phi}
				\infer[left label= \RuFR:]1{ \Gamma \seq \Delta, \lF \phi}
			\end{prooftree}
		\end{minipage}
		\begin{minipage}{0.20\textwidth}
			\begin{prooftree}
				\hypo{\Gamma \seq \Delta, \phi}
				\infer[left label= \RuNegL:]1{\neg \phi^{\anColor{u}},\Gamma \seq  \Delta}
			\end{prooftree}
		\end{minipage}
		\begin{minipage}{0.17\textwidth}
			\begin{prooftree}
				\hypo{\phi^{\anColor{u}}, \Gamma \seq \Delta}
				\infer[left label= \RuNegR:]1{ \Gamma \seq \Delta, \neg \phi}
			\end{prooftree}
		\end{minipage}
	\end{minipage}
	
	\bigskip
		
	\begin{center}{}
		\begin{minipage}{0.39\textwidth}
			\begin{prooftree}
				\hypo{\Gamma_l \seq \Delta_l,\phi}
				\hypo{\phi^{\anColor{u}},\Gamma_r \seq \Delta_r}
				\infer[left label= \RuCut:]2{\Gamma_l, \Gamma_r \seq  \Delta_l, \Delta_r}
			\end{prooftree}
		\end{minipage}	
		\begin{minipage}{0.23\textwidth}
			\begin{prooftree}
				\hypo{\Gamma \seq \Delta}
				\infer[left label=\RuWeakL:]1{\phi^{\anColor{u}},\Gamma \seq \Delta}
			\end{prooftree}
		\end{minipage}		
		\begin{minipage}{0.25\textwidth}
			\begin{prooftree}
				\hypo{\Gamma \seq \Delta}
				\infer[left label=\RuWeakR:]1{\Gamma \seq \Delta,\phi}
			\end{prooftree}
		\end{minipage}
	\end{center}
	
	\medskip
		
	\begin{center}
		\begin{minipage}{0.10\textwidth}
			\phantom{X}
		\end{minipage}
		\begin{minipage}{0.29\textwidth}
			\begin{prooftree}
				\hypo{[\Gamma \seq \Delta]^{\dx}}
				\ellipsis{}{\Gamma \seq \Delta}
				\infer[left label= \RuDischarge:]1{\Gamma \seq \Delta}
			\end{prooftree}
		\end{minipage}
		\begin{minipage}{0.23\textwidth}
			\begin{prooftree}
				\hypo{\phi^{\anColor{f}},\Gamma \seq \Delta}
				\infer[left label= \RuU:]1{ \phi^{\anColor{u}}, \Gamma \seq \Delta}
			\end{prooftree}
		\end{minipage}
		\begin{minipage}{0.25\textwidth}
			\begin{prooftree}
				\hypo{\phi^{\anColor{u}},\Gamma \seq \Delta}
				\infer[left label= \RuF:]1{ \phi^{\anColor{f}}, \Gamma \seq \Delta}
			\end{prooftree}
		\end{minipage}
		
	\end{center}

	\caption{Rules of \GKe}
	\label{fig.rulesGKe}
\end{figure}

We will read proof trees `upwards', so nodes labelled by premises are viewed as children of nodes labelled by the conclusion of a rule. A node $v$ is called an \emph{ancestor}  of a node $u$ if there are nodes $v=v_0$, \dots , $v_n=u$ with $v_{i+1}$ being a child of $v_i$ for $i = 0,...,n-1$ and $n > 0$.

\begin{definition}[Derivation]
	A \emph{\GKe derivation} $\pi = (T,P,\sfS,\sfR)$ is a quadruple such that
	$(T,P)$ is a, possibly infinite, tree with nodes $T$ and parent relation $P \subseteq T \times T$; 
	$\sfS$ is a function that maps every node $u \in T$ to an annotated sequent;
	$\sfR$ is a function that maps every node $u \in T$ to either (i) the name of a rule in Figure \ref{fig.rulesGKe} or 
	(ii) a discharge token, such that (i) the specifications of the rules in Figure \ref{fig.rulesGKe} are satisfied, (ii) every node labelled with a discharge token is a leaf, and (iii)
	for every leaf $l$ that is labelled with a discharge token $\dx \in \Tokens$ there is  an ancestor $c(l)$ of $l$ that is labelled with \RuDischarge[\dx] (and such that $l$ and $c(l)$ are labelled with the same sequent). In the latter condition, we call $l$ a \emph{repeat leaf}, and $c(l)$  \emph{companion} node of $l$.
	
	A \GKea derivation of a sequent $\Gamma \seq \Delta$ is a \GKea derivation, where the root is labelled by $\Gamma \seq \Delta$.
\end{definition}

	Let $\pi = (T,P,\sfS,\sfR)$ be a derivation. We will be working with the following two trees associated to $\pi$.
	\begin{enumerate}
		\item[(i)] The usual proof tree $\calT_\pi = (T,P)$.
		\item[(ii)] The \emph{proof tree with back edges} $\cyclicPT = (T,P^C)$ where $P^C$ is the parent relation plus back-edges for each repeat leaf, i.e.,
  \( P^C = P \cup \{(l,c(l))\mid l \text{ is a repeat leaf}\}.\) 
	\end{enumerate}

A \emph{path} in a \GKea derivation $\pi = (T,P,\sfS,\sfR)$ is a path in $\cyclicPT$.

\begin{definition}[Successful path]
	 A path $\tau$ in a \GKea derivation is called \emph{successful} if
	\begin{enumerate}
		\item every sequent on $\tau$ has a formula in focus, and
		\item  $\tau$ passes through an application of \RuFL, where the principal formula is in focus.
	\end{enumerate}

Let $v$ be a  a repeat leaf in a \GKea derivation  $\pi = (T,P,\sfS,\sfR)$ with companion  $c(v)$, and let $\tau_v$ denote the path  in $(T,P)$ from $c(v)$ to $v$. We call $v$ a \emph{discharged leaf} if the path $\tau_v$ is successful.		
	A leaf is called \emph{closed} if it is either a discharged leaf or labelled by \Ax and \emph{open}, otherwise.
\end{definition}

\begin{definition}[Proof]
	A \emph{\GKea proof} $\pi$ is a finite \GKea derivation, where every leaf is closed.
\end{definition}

The next lemma is usually a consequence of guardedness. For the \GKea proof system this is immediate.
\begin{lemma}\label{lem.guarded}
	If  $v$ is a discharged leaf in a \GKea proof, then there is a node labelled by \RuDia on $\tau_v$.
\end{lemma}

The soundness and completeness of \GKe follows from \cite{Rooduijn2021} wherein a cyclic hypersequent calculi is given for a class of modal logics with the master modality characterised by frame conditions. In the case of no frame conditions, the calculus of \cite{Rooduijn2021}, which employs the dual modalities  $\neg \ldia \neg \phi$ and $ \neg \lF \neg \phi$ as primitive, coincides with \GKe.

\begin{theorem}[Soundness and Completeness]
	There is a \GKe proof of $\Gamma \seq \Delta$ iff $\bigwedge \Gamma \rightarrow \bigvee \Delta$ is valid.
	\end{theorem}

\section{Cut Elimination}\label{sec.cutElimination}
We introduce a cut-elimination method based on reductive cut elimination but tailored to cyclic proofs. 
The cut reductions are the usual ones save for the additional care that must be taken when a discharge rule is duplicated. The reductions are listed in \autoref{sec.sub.CutReductions}.

The main difficulty in this approach are cut reductions, that alter successful paths. As an example consider the principal cut

\(
	\begin{prooftree}
		\hypo{\pi_0}
				\infer[no rule]1{\Gamma \seq \Delta, \phi, \ldia\lF \phi}
		\infer[left label= \RuFR]1{\Gamma \seq \Delta, \lF \phi}
		\hypo{\pi_1}
				\infer[no rule]1{\ldia \lF \phi^f, \Gamma \seq \Delta}
		\hypo{\pi_2}
				\infer[no rule]1{\phi^u, \Gamma\seq \Delta}
		\infer[left label= \RuFL]2{\lF \phi^f,\Gamma \seq \Delta}
		\infer[left label= \RuU]1{\lF \phi^u,\Gamma \seq \Delta}
		\infer[left label= \RuCut]2{\Gamma \seq \Delta}
	\end{prooftree}
	\)

\vspace{-1.5em}
\noindent
and its reduction to
\hfill
	\(\begin{prooftree}
		\hypo{\pi_0}
				\infer[no rule]1{\Gamma \seq \Delta, \phi, \ldia\lF \phi}
		\hypo{\pi_1}
				\infer[no rule]1{\ldia \lF \phi^f, \Gamma \seq \Delta}
		\infer[left label= \RuU]1{\ldia \lF \phi^u, \Gamma \seq \Delta}
		\infer[left label= \RuCut]2{\Gamma \seq \Delta, \phi}
		\hypo{\pi_2}
				\infer[no rule]1{\phi^u,\Gamma \seq \Delta}
		\infer[left label= \RuCut]2{\Gamma \seq \Delta}
	\end{prooftree}\)

Note that in this reduction an \RuFL rule, where the principal formula is in focus, gets removed. Hence, paths that were successful before the reduction, could be unsuccessful afterwards.
We define {`unimportant' cuts} in such a way that in all cut reductions, successful paths remain successful, i.e. the behaviour from above may not occur. Thus, for unimportant cuts it suffices to push all cuts upwards until a repeat is reached. 
The treatment of {important cuts} is more complicated, as descendants of the cut-formulas could be in focus, and therefore being critical in the discharge condition of repeats. In order to make these observations formal we first need some technical definitions.
\subsection{Preliminary definitions}\label{sec:pre}
Let $(G,E)$ be a graph. A \emph{strongly connected subgraph} of $G$ is a set of nodes $A \subseteq G$, such that from every node of $A$ there is a path to every other node in $A$. A \emph{cluster} of $G$ is a maximally strongly connected subgraph of $G$. A cluster is called \emph{trivial} if it consists of only one node. 

Let $\pi$ be a \GKe proof. A \emph{cluster} of $\pi$ is a cluster of $\cyclicPT$. Let $\calS_\pi$ be the set of non-trivial clusters of $\pi$. We define a relation $\sccr$ on $\calS_\pi$ as follows: $S_1 \sccr S_2$ if $S_1 \neq S_2$ and there are nodes $v_1 \in S_1, v_2 \in S_2$ such that there is a path from $v_1$ to $v_2$ in $\cyclicPT$. The relation $\sccr$ is irreflexive, transitive and antisymmetric. We write $\depth(S)$ for the length of the longest path in $(\calS_\pi,\sccr)$ starting from cluster $S$.

For a node $v$ in a proof $\pi$, we define the \emph{depth} of $v$ to be
\[
\depth(v)=\max\{\depth(S)\| S \in \calS_\pi \text{ and there is a path from }v \text{ to some } u \in S\}
\] 
where $\max\emptyset=0$.

The \emph{component} of $v$, written $\comp(v)$ is the set of nodes $u \in \pi$, that are reachable from $v$ in $\cyclicPT$ with $\depth(u) = \depth(v)$. Note that $\comp(v)$ does not have to coincide with a cluster in $\pi$, but may contain multiple clusters. The component of the root is called the \emph{root-component} and the cluster of the root is called the \emph{root-cluster}.

Any node $v$ in a non-trivial cluster of a \GKe proof $\pi$ has a formula in focus, as it is on the path $\tau_l$ of a discharged leaf $l$ to its companion. For nodes in a trivial cluster this is not necessarily the case. We can apply \RuF and \RuU rules in a certain way to minimise nodes with a formula in focus. By doing so, nodes with a formula in focus resemble the non-trivial clusters of the proof tree with back edges: Any node with a formula in focus is either in a non-trivial cluster or is labelled by \RuF. 

Formally, we call a \GKe proof \emph{minimally focused} if 
	(i) if $v$ is labelled by \RuU, then its child is labelled by \RuDischarge[] and 
	(ii) if $\depth(v) < \depth(v')$ for a child $v$ of $v'$, then $v$ is labelled by \RuF.
As every proof can be transformed to a minimally focused proof of the same sequent where only \RuU and \RuF rules are added and removed, we will always assume that \GKe proofs are minimally focused.

\begin{definition}[Important cut]
	Let $\sfC$ be an occurrence of a \RuCut-rule in a \GKea proof $\pi$. 
	We call $\sfC$ \emph{important}, if all formulas in the conclusion of $\sfC$ are out of focus; and \emph{unimportant}, otherwise. 
\end{definition}

\begin{definition}[Cut rank]
	The \emph{rank of a formula} $\phi$, denoted  $\rank(\phi)$, is the maximal nesting depth of $\lF$'s in $\phi$.\footnote{For example, the rank of $\lF(\ldia \lF p \land \lF \lnot \ldia \lF p)$ is $3$.}
	The \emph{rank of a \RuCut-rule} is the rank of its cut formula. 	
	The \emph{cut-rank of a proof} is the maximal rank amongst its \RuCut-rules.
\end{definition}

Let $v'$ be a child node of $v$ in a \GKe derivation $\pi$. We call a formula $\phi'$ at $v'$ an \emph{immediate descendant} of $\phi$ at $v$ if either (i) $\phi$ and $\phi'$ are designated formulas in the rule description or (ii) $\phi = \phi'$ and both are at the same side of the sequent.
A formula $\psi$ at a node $u$ is called a \emph{descendant} of a formula $\phi$ at a node $v$ if there is a path $v=v_0$, \dots, $v_n=u$ containing formulas $\phi=\phi_0$, \dots, $\phi_n=\psi$ respectively such that $\phi_{i+1}$ is an immediate descendant of $\phi_i$ for $i = 0,...,n-1$.
A descendant $\psi$ at a node $u$ of a formula $\phi$ at a node $v$ is called a \emph{component descendant} of $\phi$ if $u \in \comp(v)$.

The next lemma justifies the definition of unimportant cuts. It implies, that pushing unimportant cuts upwards does not alter successful paths. Note that this uses our assumption that \GKe proofs are minimally focused.

\begin{lemma}\label{lem.unimporantCompDescendent}
	Let $\sfC$ be an unimportant cut in a \GKea proof $\pi$. Every component descendant of the cut-formula of $\sfC$ is out of focus.
\end{lemma}

\begin{definition}[Subproof]
	For a node $v$ in $\pi$, the \emph{subproof of $\pi$ rooted at $v$}, denoted $\pi_v$, is the result of recursively replacing every open leaf $l$ in $\pi_v$  with $\pi_{c(l)}$. In order to guarantee that \RuDischarge[] rules are labelled by unique discharge tokens, discharge tokens $\dx$ are replaced by fresh discharge tokens, whenever a \RuDischarge[\dx] rule is duplicated during a reduction step.
\end{definition}
Note that $\pi_v$ is well-defined, as $v$ is a descendent of $c(l)$ for every open leaf $l$. Hence, at some point $\pi_{c(l)}$ has no open leaves.

\subsection{Important cut}

The strategy to eliminate an important cut of rank $n$ is as follows. We first push the cut upwards, until we reach a `critical' cut. Then we take $\pi_0$ and $\pi_1$, the proofs of the left and right premise of this cut, and `\emph{zip}' them together, while deleting all descendants of the cut formula. This is done recursively, similarly to how a cut can be pushed `up'. In this process we treat $\pi_0$ and $\pi_1$ differently: In $\pi_1$ the formula in focus might be the cut formula and thus may be deleted, in $\pi_0$ we keep the same formulas in focus and successful repeats in $\pi_0$ will be projected to find successful repeats in our target proof.

In order to find suitable repeats we also have to keep track of the already constructed subproof. The intermediate objects of this process are called  \emph{traversed proofs}, that correspond to proofs, where on every branch of the proof there is at most one cut of rank $n$. 

\begin{definition}[Critical cut]\label{def.cutsTidy}
	Let $\pi$ be a \GKea proof and $\sfC$ an important cut in $\pi$. We call $\sfC$ \emph{critical} if the cut formula of $\sfC$ is of the form $\lF \phi$.
\end{definition}

\begin{lemma}\label{lem.cutsTidy}
	Let $\pi$ be a \GKe proof of cut-rank $n$ where the only cut of rank $n$ is important and at the root. Then there is a \GKea proof $\pi'$ of the same sequent with cut-rank $n$, where all cuts of rank $n$ are critical.
\end{lemma}
\begin{proof}
	We can apply cut reductions from Appendix \ref{sec.sub.CutReductions} until all cuts of rank $n$ are critical. In all cut reductions, except where $\lF \phi$ is principal, the syntactic size (i.e., the number of symbols) of the cut formula is not increased, and in the reduction for $\ldia$ the syntactic size of the cut formula decreases. Due to Lemma \ref{lem.guarded} there is a \RuDia rule on the path from a companion node to all of its discharged leaves. Hence, the syntactic size of the cut formula $\psi$ decreases until $\psi$ is of the form $\lF \phi$.
\end{proof}

\begin{definition}[Traversed proof]
	An \emph{$\lF\phi$-traversed proof} $\rho$ of $\Sigma \seq \Pi$ is a derivation of $\Sigma \seq \Pi$, where all leaves are either closed or labelled by a sequent $\Gamma_l,\Gamma_r \seq \Delta_l,\Delta_r$ together with a triple $(\pi_l,\psi,\pi_r)$, written as $\merge{\pi_l}{\psi}{\pi_r}$, such that $\psi \equiv \lF\phi$ or $\psi \equiv \ldia\lF\phi$,  $\pi_l$ is a \GKe proof of $\Gamma_l\seq \Delta_l, \psi$ and $\pi_r$ is a \GKe proof of $\psi^a, \Gamma_r'\seq \Delta_r$, where $(\Gamma'_r)^u = \Gamma_r$. 
	If $\lF\phi$ is clear from the context we will just write \emph{traversed proof}.
\end{definition}
By applying a cut with cut-formula $\psi$ at every open leaf, every $\lF \phi$-traversed proof $\rho$ corresponds to a \GKe proof $\pi$, where on every branch of the proof there is at most one cut of rank $\rank(\lF \phi)$. Hence, transforming a $\lF \phi$-traversed proof to a traversed proof without open leaves corresponds to eliminating cuts of rank $\rank(\lF \phi)$.

\begin{lemma}\label{lem.cutsImportant}
	Let $\pi$ be a \GKe proof with a critical cut of rank $n+1$ at the root, i.e. $\pi$ is of the form
	\begin{align*}
		\begin{prooftree}
			\hypo{\pi_0}
			\infer[no rule]1{\Sigma_0 \seq \Pi_0, \lF \phi}
			\hypo{\pi_1}
			\infer[no rule]1{\lF \phi^u, \Sigma_1 \seq \Pi_1}
			\infer[left label= \RuCut]2{\Sigma_0, \Sigma_1 \seq \Pi_0, \Pi_1}
		\end{prooftree}
	\end{align*}	
	where $\pi_0$ and $\pi_1$ are of cut-rank $n$, and the cut formula $\lF\phi$ has rank $n+1$. Then we can construct a proof $\pi'$ of $\Sigma_0,\Sigma _1 \seq \Pi_0,\Pi_1$ of cut-rank $n$.
\end{lemma}
\begin{proof}
	Let $\rho_I$ be the $\lF\phi$-traversed proof of $\Sigma_0,\Sigma _1 \seq \Pi_0,\Pi_1$ consisting of an open leaf labelled by $\Sigma_0,\Sigma _1 \seq \Pi_0,\Pi_1$ together with $\merge{{\pi}_0}{\psi}{{\pi}_1}$, this we will denote by
	\begin{align*}
		\begin{prooftree}
			\hypo{\merge{{\pi}_0}{\psi}{{\pi}_1}}
			\infer[no rule]1{\Sigma_0,\Sigma_1 \seq \Pi_0,\Pi_1}
		\end{prooftree}
	\end{align*}
	Starting from $\rho_I$ we will inductively transform the traversed proof, until we end up with a traversed proof of $\Sigma_0,\Sigma_1 \seq \Pi_0,\Pi_1$, where all leaves are closed. This will be done as follows. Let $\rho$ be a traversed proof. If all leaves are closed, we are done. Otherwise, consider the leftmost open leaf $v$ labelled by 
	\begin{align*}
		\begin{prooftree}
			\hypo{\merge{\pi_l}{\psi}{\pi_r}}
			\infer[no rule]1{\Gamma_l,\Gamma_r \seq \Delta_l,\Delta_r}
		\end{prooftree}
	\end{align*}
	Let the roots of $\pi_l$ and $\pi_r$ be labelled by the rules $\Ru_l$ and $\Ru_r$, respectively. We transform $\rho$ by a case distinction on $\Ru_l$ and $\Ru_r$. We only show the crucial cases, the full proof can be found in Appendix \ref{sec.appendixImportantCuts}.
	\begin{itemize}	
		\item If $\Ru_l$ is \RuDischarge, then $\pi_l$ has the form
		\begin{align*}
			\begin{prooftree}
				\hypo{\pi_l'}
				\infer[no rule]1{\Gamma_l \seq \Delta_l, \psi}
				\infer[left label= \RuDischarge]1{\Gamma_l \seq \Delta_l,\psi}
			\end{prooftree}
		\end{align*}
		We proceed by making the following case distinction.
		\begin{itemize}
			\item If there is a node $c$ in $\rho$, that is an ancestor of $v$ and labelled by $\Gamma_l,\Gamma_r \seq \Delta_l,\Delta_r$, such that the path from $c$ to $v$ is successful, then insert a \RuDischarge[\dz] rule at $c$ and let $v$ be the discharged leaf $[\Gamma_l,\Gamma_r \seq \Delta_l,\Delta_r]^\dz$ with fresh discharge token $\dz$. 
			
			\item Else, unfold $\pi_l$, i.e. let $\tilde{\pi}_l$ be the proof obtained from $\pi_l'$ by replacing every discharged leaf labelled by $\dx$ with $\pi_l$.
			We replace $v$ by
			\begin{align*}
				\begin{prooftree}
					\hypo{\merge{\tilde{\pi}_l}{\psi}{\pi_r}}
					\infer[no rule]1{\Gamma_l,\Gamma_r \seq \Delta_l,\Delta_r}
				\end{prooftree}
			\end{align*}
		\end{itemize}
	
		\item If $\Ru_r$ is \RuDischarge[\dy], then unfold $\pi_r$.

		\item If $\Ru_l$ is $\RuFR$ and $\Ru_r$ is $\RuFL$, and $\pi_l$ and $\pi_r$ are of the forms
			\begin{align*}
					\begin{prooftree}
					\hypo{\pi_l'}
					\infer[no rule]1{\Gamma_l \seq \Delta_l, \phi, \ldia\lF\phi}
					\infer[left label= \RuFR]1{\Gamma_l \seq \Delta_l,\lF\phi}
				\end{prooftree}
				&\qquad
					\begin{prooftree}
					\hypo{\pi_r'}
					\infer[no rule]1{\ldia \lF \phi^a, \Gamma_r \seq \Delta_r}
					\hypo{\pi_r^u}
					\infer[no rule]1{\phi^u, \Gamma_r \seq \Delta_r}
					\infer[left label= \RuFL]2{\lF\phi^a,\Gamma_r \seq  \Delta_r}
				\end{prooftree}
			\end{align*}
			then $v$ is replaced by
			\begin{align*}
				\begin{prooftree}
					\hypo{\merge{\pi_l'}{\ldia\lF\phi}{\pi_r'}}
					\infer[no rule]1{\Gamma_l,\Gamma_r \seq \Delta_l,\Delta_r,\phi}
					\hypo{\pi_r^u}
					\infer[no rule]1{\phi^u, \Gamma_r \seq \Delta_r}
					\infer[left label= \RuCut]2{\Gamma_l,\Gamma_r \seq \Delta_l,\Delta_r}
				\end{prooftree}
			\end{align*}
			where the introduced cut has rank $n$.

		\item If $\Ru_l$ is \RuU or \RuF of the form 
		\begin{align*}
			\begin{prooftree}
				\hypo{\pi_l'}
				\infer[no rule]1{\Gamma_l' \seq \Delta_l, \psi}
				\infer[left label= $\Ru_l$]1{\Gamma_l \seq \Delta_l, \psi}
			\end{prooftree}
		&
	\qquad\qquad\text{then $v$ is replaced by} 
		&
			\begin{prooftree}
				\hypo{\merge{\pi_l'}{\psi}{\pi_r}}
				\infer[no rule]1{\Gamma_l',\Gamma_r \seq \Delta_l,\Delta_r}
				\infer[left label= $\Ru_l$]1{\Gamma_l,\Gamma_r \seq \Delta_l,\Delta_r}
			\end{prooftree}
	\end{align*}
		
		\item If $\Ru_r$ is $\RuU$ or $\RuF$ of the form
		\begin{align*}
			\begin{prooftree}
				\hypo{\pi_r'}
				\infer[no rule]1{\psi^b, \Gamma_l' \seq \Delta_l}
				\infer[left label= $\Ru_r$]1{\psi^a, \Gamma_l \seq \Delta_l}
			\end{prooftree}
		&
			\qquad\qquad\text{then $v$ is replaced by} 
		&
			\begin{prooftree}
				\hypo{\merge{\pi_l}{\psi}{\pi_r'}}
				\infer[no rule]1{\Gamma_l,\Gamma_r \seq \Delta_l,\Delta_r}
			\end{prooftree}
		\end{align*}

		\item If $\Ru_l$ or $\Ru_r$ is any non-principal rule, we push the open leaf `up', similar to how a cut is pushed away from the root.
		
	\end{itemize}
	
	Let $\rho_i$ and $\rho_j$ be traversed proofs. We write $\rho_i < \rho_j$, if $\rho_j$  can be obtained from $\rho_i$ by the above construction and $\rho_i \neq \rho_j$. 
	If $\rho_i < \rho_j$, then $\rho_i$ is a subproof of $\rho_j$, in the sense that $\rho_j$ can be obtained from $\rho_i$ by replacing some open leaves in $\rho_i$ by traversed proofs and inserting nodes labelled by \RuDischarge[]. Thus, $\rho_j$ consists of at least the nodes in $\rho_i$ and we can identify nodes in $\rho_i$ with nodes in $\rho_j$.
	
	\noindent
	Let $\rho > \rho_I$ and let $v$ be an open leaf in $\rho$ labelled by
	\begin{align*}
		\begin{prooftree}
			\hypo{\merge{\pi_l}{\psi}{\pi_r}}
			\infer[no rule]1{\Gamma_l,\Gamma_r \seq \Delta_l,\Delta_r}
		\end{prooftree}
	\end{align*}
	We can define a node $u_0(v) \in \pi_0$ with $\pi_l = \pi_{u_0(v)}$, $\sfS(u_0(v)) = \Gamma_l \seq \Delta_l, \psi$ and an analogous node $u_1(v) \in \pi_1$. These definitions extend to nodes $w$ below an open leaf $v$. Hence, nodes below an open leaf can be labelled by only finitely many distinct sequents.
	
	Let $\tau$ be the path from the root of $\rho$ to the open leaf $v$. Using the above definitions of $u_0(w)$ and $u_1(w)$ we obtain corresponding paths $\tau_0$ in $\pi_0$ and $\tau_1$ in $\pi_1$.

	Combining these observations, we can find a bound on the height of open leaves: As there are only finitely many distinct sequents on $\tau$, at some point an open leaf $v$ is reached, such that an ancestor $c$ in $\rho$ is labelled by the same sequent. The path $\tau(v)$ from $c$ to $v$ in $\rho$ corresponds to a path $\tau_0(v)$ in $\pi_0$. If $\tau(v)$ is of some certain length the path $\tau_0(v)$ is successful. This implies that the path $\tau(v)$ is successful as well and the leaf $v$ will be closed. A more detailed proof with explicit bounds can be found in Appendix \ref{sec.appendixImportantCuts}.
	
	In conclusion, as every constructed tree is finitely branching, after finitely many steps, a traversed proof $\rho_T$ without open leaves is constructed.	
	As every traversed proof without open leaves is a \GKe proof, we are done.
\end{proof}

\begin{example}
	Consider the following proof $\pi$ containing a critical cut
	\begin{align*}
		\begin{prooftree}
			\hypo{\pi_0}
			\infer[no rule]1{\lF\ldia\ldia p^u \seq \lF \ldia p}
			\hypo{\pi_1}
			\infer[no rule]1{\lF\ldia p^u \seq \lF p}
			\infer[left label= \RuCut]2{\lF\ldia\ldia p^u \seq \lF p}
		\end{prooftree}
	\end{align*}
	where the two proofs $\pi_0$ and $\pi_1$ are
	\begin{align*}
		\begin{minipage}{0.54\textwidth}
			\begin{prooftree}
				\hypo{[\lF\ldia\ldia p^f \seq \lF \ldia p]^\dx}
				\infer[left label= \RuDia]1{\ldia\lF\ldia\ldia p^f \seq \ldia\lF \ldia p}
				\hypo{}
				\infer[double]1{\ldia\ldia p^u \seq \ldia\lF \ldia p}
				\infer[left label= \RuFL]2{\lF\ldia\ldia p^f \seq \ldia p, \ldia\lF \ldia p}
				\infer[left label=\RuFR]1{\mathllap{a:\qquad\qquad}\lF\ldia\ldia p^f \seq \lF \ldia p}
				\infer[left label= \RuDischarge[\dx]]1{\lF\ldia\ldia p^f \seq \lF \ldia p}		
				\infer[left label= \RuU]1{\lF\ldia\ldia p^u \seq \lF \ldia p}
			\end{prooftree}
		\end{minipage}
		\begin{minipage}{0.45\textwidth}
			\begin{prooftree}
				\hypo{[\lF\ldia p^f \seq \lF  p]^\dy}
				\infer1[\RuDia]{\ldia\lF\ldia p^f \seq \ldia\lF  p}
				\hypo{}
				\infer[double]1{\ldia p^u \seq \ldia\lF  p}
				\infer[left label= \RuFL]2{\mathllap{b:\qquad\quad}\lF\ldia p^f \seq p, \ldia\lF  p}
				\infer[left label= \RuFR]1{\lF\ldia p^f \seq \lF  p}
				\infer[left label= \RuDischarge[\dy]]1 {\lF\ldia p^f \seq \lF  p}
				\infer[left label= \RuU]1{\lF\ldia p^u \seq \lF p}
			\end{prooftree}
		\end{minipage}
	\end{align*}
	For the sake of readability we apply weakening implicitly and omit \RuWeak rules in this example. The leaves marked by double lines can be proved easily without cuts or repeats.
	
	Following the construction and notation of Lemma \ref{lem.cutsImportant} we define $\rho_I$ as above.
	In the first steps, the transformations for \RuU in $\pi_0$ and $\pi_1$ are applied, then the resulting subproofs $\pi_0'$ and $\pi_1'$ get unfolded. Adding the \RuFR rule in $\pi_1$ to $\rho_I$, results in the following traversed proof:
	\begin{align*}
		\begin{prooftree}
			\hypo{\merge{{\pi}_a}{\lF\ldia p}{{\pi}_b}}
			\infer[no rule]1{\lF\ldia\ldia p^f \seq p, \ldia\lF p}
			\infer[left label= \RuFR]1{\lF\ldia\ldia p^f \seq \lF p}
			\infer[left label= \RuU]1{\lF\ldia\ldia p^u \seq \lF p}
		\end{prooftree}
	\end{align*}
	Now the rules $\Ru_l$ and $\Ru_r$ are \RuFR and \RuFL, respectively. Hence a cut with cut-formula $\ldia p$ is introduced. Then the traversed proof gets further transformed into
	\begin{align*}
		\begin{prooftree}
			\hypo{\merge{{\pi}_0'}{\lF\ldia p}{{\pi}_1'}}
			\infer[no rule]1[]{\mathllap{v:\qquad\qquad}\lF\ldia\ldia p^f \seq \lF p}
			\infer[left label= \RuDia]1{\ldia\lF\ldia\ldia p^f \seq \ldia\lF p}
			\hypo{}
			\infer[double]1{\ldia \ldia p^u, \seq \ldia\lF p}
			\infer[left label= \RuFL]2{\lF\ldia\ldia p^f \seq \ldia p, \ldia\lF p}
			\hypo{}
			\infer[double]1{\ldia p^u \seq \ldia\lF p}
			\infer[left label= \RuCut]2{\lF\ldia\ldia p^f \seq p, \ldia\lF p}
			\infer[left label= \RuFR]1{\mathllap{c:\qquad\qquad}\lF\ldia\ldia p^f \seq \lF p}
			\infer[left label= \RuU]1{\lF\ldia\ldia p^u \seq \lF p}
		\end{prooftree}
	\end{align*}
	The root of $\pi_0'$ is labelled by \RuDischarge and the repeat condition for $v$ is satisfied: The node $c$ is labelled by the same sequent as $v$ and the path from $c$ to $v$ is successful. Thus a \RuDischarge[\dz] rule is inserted (with a fresh discharge token \dz) at $c$ and $v$ is discharged by $\dz$, which results in a \GKe proof.
\end{example}

\subsection{Unimportant cuts}

Recall the set of component nodes of a node $v$ in a derivation, written $\comp(v)$, defined in \autoref{sec:pre}.
	
\begin{definition}
	Let $\pi$ be a \GKea derivation and $v$ be a node in $\pi$. The \emph{infinite unfolding of $\comp(v)$} in $\pi$, written $\unf[v]{\pi}$, is obtained from $\pi$ by recursively replacing every discharged leaf $l$, that is a component descendant of $v$, by $\pi_l$ and removing nodes labelled by \RuDischarge[\dx] whenever no discharged leaf is labelled by $\dx$.
\end{definition}

\begin{lemma}\label{lem.cutsUnimportant}
	Let $\pi$ be a \GKea proof of cut-rank $n$, such that all cuts of rank $n$ are unimportant and in the root-cluster. Then we can transform $\pi$ into a \GKea proof $\pi'$ of the same sequent with cut-rank $\leq n$, where all cuts of rank $n$ are important.
\end{lemma}
\begin{proof}
	Let $\Gamma \seq \Delta$ be the sequent at the root $r$ of $\pi$. For formulas $\phi$ and $\psi$, we write $\phi \tracestep \psi$ if 
	$\psi$ is either a direct subformula of $\phi$ or $\phi \equiv \lF \chi$ and $\psi \equiv \ldia \lF \chi$. 		%
	The \emph{closure} $\Clos(\Gamma\seq \Delta)$ of a sequent $\Gamma \seq \Delta$
	is the least superset of $\Gamma \cup \Delta$ that is closed under $\tracestep$.  Let $m = |\Clos(\Gamma \seq \Delta)|$. 

Let $\unf[r]{\pi}$ be the infinite unfolding of $\comp(r)$ of $\pi$ and let $A$ be $\comp(r)$ in $\unf[r]{\pi}$. Note that $A$ is infinite.
	We want to push the cuts (also the ones with cut-rank $< n$) occurring in $A$ upwards in $\unf[r]{\pi}$, until the set of nodes $A_0$ of height $8^m+1$ in $A$ does not contain cuts. 
	
	Choose a root-most node $v$ labelled by \RuCut, take a lowest descendant $v'$ of $v$, that is labelled by \RuCut, such that both premises are not labelled by \RuCut(that could also be $v$ itself). This can always be found as $\pi$ is a proof, hence a repeat cannot consist of \RuCut-rules only. Apply a suitable cut reduction defined in Appendix \ref{sec.sub.CutReductions} to $v'$. At some point this will be a non-principal cut reduction, because $\pi$ is a proof and the cut-formula is out of focus: There are infinitely many \RuFL rules, where the principal formula is in focus on every infinite path in $A$. Thus the height of the cut-free sub-proof increases. Continue this process until $A_0$ is cut-free. All cuts of rank $n$, that where pushed outside of $A$, are important: If the cut is pushed out of the cluster all formulas in the conclusion become out of focus (as $\pi$ is minimally focused) and no cut reduction puts formulas in focus again. Because all cut-reductions defined in Appendix \ref{sec.sub.CutReductions} preserve the cut-rank, all cuts have cut-rank $\leq n$.
	
	As $A_0$ is cut-free, each formula in $A_0$ is in $\Clos(\Gamma \seq \Delta)$, where it could occur or not occur in $\Gamma$ and $\Delta$, and be in focus or out of focus in $\Gamma$. Thus, there are at most $8^m$ many distinct sequents in $A_0$ and on each branch in $A_0$ there is a node $v$ such that an ancestor $c(v)$ of $v$ is labelled by the same sequent. For each such branch choose the root-most such node $v$, insert a \RuDischarge[\dx] rule at $c(v)$ and let $v$ be a discharged leaf labelled by $\dx$ with fresh discharge token $\dx$.
	All sequents in $A_0$ have a formula in focus. Because of Lemma \ref{lem.unimporantCompDescendent} none of the cut-reductions affect formulas in focus,  and \RuFL-rules where the principal formula is in focus, remain.  Hence, the path from $c(v)$ to $v$ is successful and we obtain a proof with cut-rank $n$, where all cuts of rank $n$ are important.
\end{proof}

\begin{example}
	Consider the following proof $\pi$ containing one unimportant cut.
	\begin{align*}
		\begin{prooftree}
			\hypo{[\lF \ldia p^f \seq \lF p]^\dx}
			\infer[left label= \RuDia]1{\ldia\lF \ldia p^f \seq \ldia\lF p}
			\hypo{}
			\infer[double]1{\ldia p^u \seq \ldia\lF p}
			\infer[left label= \RuFL]2{\lF\ldia p^f \seq \ldia \lF p}
			\hypo{}
			\infer[double]1
			{\ldia \lF p^u \seq \lF p}
			\infer[left label= \RuCut]2{\lF \ldia p^f \seq \lF p}
			\infer[left label= \RuDischarge[\dx]]1{\lF \ldia p^f \seq \lF p}
		\end{prooftree}
	\end{align*}
	For the sake of readability, we apply weakening implicitly and omit \RuWeak rules in this example. The leaves marked by double lines are in the next cluster and their (straightforward) proofs will be omitted.	Following the construction described in Lemma \ref{lem.cutsUnimportant} we push the cut upwards using the cut reductions from Appendix \ref{sec.sub.CutReductions}. Here we only unfold the discharged leaf once, as it suffices in this example, in contrast to working with $\unf[r]{\pi}$ as in the proof of Lemma \ref{lem.cutsUnimportant}.
	First, a non-principal cut reduction for \RuFL is applied:
	\begin{align*}
		\begin{prooftree}
			\hypo{\pi}
			\infer[no rule]1{\lF \ldia p^f \seq \lF p}
			\infer[left label= \RuDia]1{\ldia\lF \ldia p^f \seq \ldia\lF p}
			\hypo{}
			\infer[double]1{\ldia \lF p^u \seq p, \ldia\lF p}
			\infer[left label= \RuFR]1{\ldia \lF p^u \seq \lF p}
			\infer[left label= \RuCut]2{\ldia\lF\ldia p^f \seq \lF p}
			\hypo{}
			\infer[double]1{\ldia p^u \seq \ldia \lF p}
			\hypo{}
			\infer[double]1{\ldia\lF p^u \seq \lF p}
			\infer[left label= \RuCut]2{\ldia p^u \seq \lF p}
			\infer[left label= \RuFL]2{\lF \ldia p^f \seq \lF p}			
		\end{prooftree}
	\end{align*}
	The cut in the right subproof is pushed out of the cluster. In the left subproof, principal cut-reductions for \RuFR and \RuDia are applied:
\begin{align*}
	\begin{prooftree}
		\hypo{\pi}
		\infer[no rule]1{\lF \ldia p^f \seq \lF p}
		\infer[]1[\RuDia]{\ldia\lF \ldia p^f \seq \ldia\lF p}
		\hypo{}
		\infer[double]1{\ldia \lF p^u \seq \lF p}
		\infer[left label= \RuCut]2{\mathllap{v:\qquad\qquad}\lF\ldia p^f \seq \lF p}
		\infer[left label=\RuDia]1{\ldia\lF\ldia p^f \seq p,\ldia\lF p}
		\infer[left label= \RuFR]1{\ldia\lF\ldia p^f \seq \lF p}
		\hypo{}
		\infer[double]1{\ldia p^u \seq \ldia \lF p}
		\hypo{}
		\infer[double]1{\ldia\lF p^u \seq \lF p}
		\infer[left label= \RuCut]2{\ldia p^u \seq \lF p}
		\infer[left label= \RuFL]2{\mathllap{c:\qquad}\lF \ldia p^f \seq \lF p}			
	\end{prooftree}
\end{align*}
The node $v$ below the cut is labelled by the same sequent as its ancestor $c$ and the path from $c$ to $v$ is successful. Thus a \RuDischarge[\dy] is inserted at $c$ and $v$ is discharged by $\dy$. This results in a \GKea proof, where the only remaining cut is important.	
\end{example}

\subsection{Main theorem}

We want to prove cut-elimination for $\GKe$ by induction on the cut-rank, hence it suffices to transform any proof $\pi$ of cut-rank $n+1$ to a proof $\pi'$ of cut-rank $n$.

\begin{lemma}[Reduction Lemma]\label{lem.cutReduction}
	Let $\pi$ be a proof of cut-rank $n+1$. Then we can transform $\pi$ into a \GKe proof $\pi'$ of cut-rank $n$ of the same sequent. 
\end{lemma}
\begin{proof}
	By induction on the number of clusters in $\pi$ with unimportant cuts of rank $n+1$ with a sub-induction on the number of clusters in $\pi$ with important cuts of rank $n+1$. 
	
	Let $\pi_0$ be a subproof of $\pi$, where all cuts with rank $n+1$ are in the root-cluster. If the root-cluster is trivial, there is one important cut at the root of $\pi_0$. Otherwise, all cuts in the root-cluster of $\pi_0$ are unimportant.
	In the first case, Lemma \ref{lem.cutsTidy} transforms $\pi_0$ to $\pi_0'$, where all cuts of rank $n+1$ are critical. Those can then be reduced using Lemma \ref{lem.cutsImportant}, which yields a proof $\pi_1$ with cut-rank $n$.
	In the second case, Lemma \ref{lem.cutsUnimportant} yields a proof $\pi_1$ with cut-rank $n+1$, where all cuts of rank $n+1$ are important. 
	In both cases, we substitute $\pi_0$ by $\pi_1$ in $\pi$ and obtain a proof $\pi'$, where we can apply the induction hypothesis. 
\end{proof}

\begin{theorem}[Cut elimination]\label{thm.cutElimination}
	We can transform every \GKea proof $\pi$ into a cut-free \GKea proof $\pi'$ of the same sequent.
\end{theorem}
\begin{proof}
	By induction on the cut-rank of $\pi$ using Lemma \ref{lem.cutReduction}.
\end{proof}

\section{Conclusion}
We have presented a syntactic cut-elimination procedure for the cyclic proof system \GKe. To our knowledge, this is the first cut-elimination method that works directly on cyclic proofs without a detour through infinitary proofs. 
We anticipate that the introduced approach generalises to other annotated, cyclic proof systems for fragments of the modal $\mu$-calculus. Two natural candidates are Propositional Dynamic Logic and the alternation-free $\mu$-calculus. The former has a similar modal structure as \MLe; the latter is known to admit a focused cyclic proof system~\cite{Marti2021}. In both cases important and unimportant cuts can be defined as in this paper, allowing an analogous cut elimination process for unimportant cuts. Complications arise in the treatment of important cuts, especially if formulas in focus occur in multiple premises of a rule or multiple formulas in a sequent are in focus. A potential mitigation would be to employ a multi-cut rule.

\bibliographystyle{eptcs}
\bibliography{FICS.bib}

\clearpage
\appendix
\setcounter{theorem}{0}
\renewcommand\thetheorem{\thesection.\arabic{theorem}}

\section{Appendix: Cut Reductions}\label{sec.sub.CutReductions}
For readability we state the cut reductions for a simplified \RuCut rule, where $\Gamma_l = \Gamma_r$ and $\Delta_l = \Delta_r$; this can be generalised in the obvious way. 

\subsection{Principal cut reductions}
	%
	\begin{multline*}
		\begin{prooftree}
			\hypo{\pi_0}
			\infer[no rule]1{\Gamma \seq \Delta, \phi, \ldia\lF \phi}
			\infer[left label= \RuFR]1{\Gamma \seq \Delta, \lF \phi}
			\hypo{\pi_1}
			\infer[no rule]1{\ldia \lF \phi^u, \Gamma \seq \Delta}
			\hypo{\pi_2}
			\infer[no rule]1{\phi^u, \Gamma\seq \Delta}
			\infer[left label= \RuFL]2{\lF \phi^u,\Gamma \seq \Delta}
			\infer[left label= \RuCut]2{\Gamma \seq \Delta}
		\end{prooftree}	\\
		\longrightarrow \qquad
		\begin{prooftree}
			\hypo{\pi_0}
			\infer[no rule]1{\Gamma \seq \Delta, \phi, \ldia\lF \phi}
			\hypo{\pi_1}
			\infer[no rule]1{\ldia \lF \phi^u, \Gamma \seq \Delta}
			\infer[left label= \RuCut]2{\Gamma \seq \Delta, \phi}
			\hypo{\pi_2}
			\infer[no rule]1{\phi^u,\Gamma \seq \Delta}
			\infer[left label= \RuCut]2{\Gamma \seq \Delta}
		\end{prooftree}
	\end{multline*}
 \smallskip
 
\[
		\begin{prooftree}
			\hypo{\pi_0}
			\infer[no rule]1{\gamma^a \seq \Delta, \phi}
			\infer[left label= \RuDia]1{\ldia \gamma^a \seq \ldia \Delta, \ldia \phi}
			\hypo{\pi_1}
			\infer[no rule]1{\phi^u \seq \Delta}
			\infer[left label= \RuDia]1{\ldia \phi^u \seq \ldia \Delta}
			\infer[left label= \RuCut]2{\ldia \gamma^a \seq \ldia \Delta}
		\end{prooftree}	 \qquad\longrightarrow \qquad
		\begin{prooftree}
			\hypo{\pi_0}
			\infer[no rule]1{\gamma^a \seq \Delta, \phi}
			\hypo{\pi_1}
			\infer[no rule]1{\phi^u \seq \Delta}
			\infer[left label= \RuCut]2{\gamma^a \seq \Delta}
			\infer[left label= \RuDia]1{\ldia \gamma^a \seq \ldia \Delta}
		\end{prooftree}
	\]
	
	\begin{multline*}
		\begin{prooftree}
			\hypo{\pi_0}
			\infer[no rule]1{\Gamma \seq \Delta, \phi}
			\hypo{\pi_1}
			\infer[no rule]1{\Gamma \seq \Delta, \psi}
			\infer[left label= \RuAndR]2{\Gamma \seq \Delta, \phi \land \psi}
			\hypo{\pi_2}
			\infer[no rule]1{\phi^u, \psi^u, \Gamma\seq \Delta}
			\infer[left label= \RuAndL]1{\phi \land \psi^u,\Gamma \seq \Delta}
			\infer[left label= \RuCut]2{\Gamma \seq \Delta}
		\end{prooftree}	\\
		\longrightarrow \qquad
		\begin{prooftree}
			\hypo{\pi_0}
			\infer[no rule]1{\Gamma \seq \Delta, \phi}
			\hypo{\pi_1}
			\infer[no rule]1{\Gamma \seq \Delta, \psi}
			\hypo{\pi_2}
			\infer[no rule]1{\phi^u, \psi^u, \Gamma \seq \Delta}
			\infer[left label= \RuCut]2{\phi^u, \Gamma \seq \Delta}
			\infer[left label= \RuCut]2{\Gamma \seq \Delta}
		\end{prooftree}
	\end{multline*}
	
\[
		\begin{prooftree}
			\hypo{\pi_0}
			\infer[no rule]1{\phi^u, \Gamma \seq \Delta}
			\infer[left label= \RuNegR]1{\Gamma \seq \Delta, \neg \phi}
			\hypo{\pi_1}
			\infer[no rule]1{\Gamma \seq \Delta, \phi}
			\infer[left label= \RuNegL]1{\neg \phi^u,\Gamma \seq \Delta}
			\infer2[\RuCut]{\Gamma \seq \Delta}
		\end{prooftree}	\qquad \longrightarrow \qquad
		\begin{prooftree}
			\hypo{\pi_1}
			\infer[no rule]1{\Gamma \seq \Delta, \phi}
			\hypo{\pi_0}
			\infer[no rule]1{\phi^u, \Gamma \seq \Delta}
			\infer[left label= \RuCut]2{\Gamma \seq \Delta}
		\end{prooftree}
\]
\subsection{Trivial principal cut reductions}
\[
		\begin{prooftree}
			\hypo{\pi_0}
			\infer[no rule]1{\Gamma \seq \Delta, p}
			\hypo{p^u \seq p}
			\infer[left label= \RuCut]2{\Gamma \seq \Delta, p}
		\end{prooftree}
		\qquad \longrightarrow \qquad
		\begin{prooftree}
			\hypo{\pi_0}
			\infer[no rule]1{\Gamma \seq \Delta,p}
		\end{prooftree}
\]
	
\[
		\begin{prooftree}
			\hypo{p^u \seq p}
			\hypo{\pi_1}
			\infer[no rule]1{p^u, \Gamma \seq \Delta}
			\infer[left label= \RuCut]2{p^u, \Gamma \seq \Delta}
		\end{prooftree}
		\qquad \longrightarrow \qquad
		\begin{prooftree}
			\hypo{\pi_1}
			\infer[no rule]1{p^u, \Gamma \seq \Delta}
		\end{prooftree}
\]
	
\[
		\begin{prooftree}
			\hypo{\pi_0}
			\infer[no rule]1{\Gamma \seq \Delta}
			\infer[left label= \RuWeakR]1{\Gamma \seq \Delta, \phi}
			\hypo{\pi_1}
			\infer[no rule]1{\phi^u,\Gamma \seq \Delta}
			\infer[left label= \RuCut]2{\Gamma \seq \Delta}
		\end{prooftree} \qquad \longrightarrow \qquad
		\begin{prooftree}
			\hypo{\pi_0}
			\infer[no rule]1{\Gamma \seq \Delta}
		\end{prooftree}
\]
	
	\[
		\begin{prooftree}
			\hypo{\pi_0}
			\infer[no rule]1{\Gamma \seq \Delta, \phi}
			\hypo{\pi_1}
			\infer[no rule]1{\Gamma \seq \Delta}
			\infer[left label= \RuWeakL]1{\phi^u,\Gamma \seq \Delta}
			\infer[left label= \RuCut]2{\Gamma \seq \Delta}
		\end{prooftree}	\qquad \longrightarrow \qquad
		\begin{prooftree}
			\hypo{\pi_1}
			\infer[no rule]1{\Gamma \seq \Delta}
		\end{prooftree}
	\]
\subsection{Cut reductions for \RuDischarge[] , \RuU and \RuF}
 We push \RuU and \RuF rules `upwards' away from the root and unfold \RuDischarge[] rules. The presented reductions are analogous, if the right premise of the cut is labelled by \RuDischarge[], \RuU  or \RuF.

	\[
	\begin{prooftree}
		\hypo{\pi_0}
		\infer[no rule]1{\Gamma \seq \Delta, \phi}
		\infer[left label= \RuDischarge]1{\mathllap{v:\qquad\quad}\Gamma\seq \Delta, \phi}
		\hypo{\pi_1}
		\infer[no rule]1{\phi^u, \Gamma\seq \Delta}
		\infer[left label= \RuCut]2{\Gamma \seq \Delta}
	\end{prooftree}
	\qquad \longrightarrow \qquad
	\begin{prooftree}
		\hypo{\pi_0'}
		\infer[no rule]1{\Gamma \seq \Delta, \phi}
		\hypo{\pi_1}
		\infer[no rule]1{\phi^u, \Gamma\seq \Delta}
		\infer[left label= \RuCut]2{\Gamma \seq \Delta}
	\end{prooftree}
	\]
	where $\pi_0'$ is obtained from $\pi_0$ by replacing every discharged leaf labelled by $\dx$ with $\pi_v$, where $v$ is the left premise of the \RuCut rule.\footnote{Here and in the following cut reductions discharge tokens $\dy$ are replaced by fresh discharge tokens, whenever a \RuDischarge[\dy] rule is duplicated.}

\[
		\begin{prooftree}
			\hypo{\pi_0}
			\infer[no rule]1{\Gamma' \seq \Delta, \phi}
			\infer[left label= \RuDischarge]1{\Gamma'\seq \Delta, \phi}
			\infer[left label= \RuU]1{\mathllap{v:\qquad\quad}\Gamma \seq \Delta, \phi}
			\hypo{\pi_1}
			\infer[no rule]1{\phi^u, \Gamma\seq \Delta}
			\infer[left label= \RuCut]2{\Gamma \seq \Delta}
		\end{prooftree}
		\qquad \longrightarrow \qquad
		\begin{prooftree}
			\hypo{\pi_0'}
			\infer[no rule]1{\Gamma \seq \Delta, \phi}
			\hypo{\pi_1}
			\infer[no rule]1{\phi^u, \Gamma\seq \Delta}
			\infer[left label= \RuCut]2{\Gamma \seq \Delta}
		\end{prooftree}
\]
	where $\pi_0'$ is obtained from $\pi_0$ by (i) unfocusing sequents up to \RuDischarge[] rules and leaves labelled by $\dx$ and (ii) replacing every discharged leaf labelled by $\dx$ with subproof $\pi_v$, where $v$ is the left premise of the \RuCut rule.

\[
\begin{prooftree}
	\hypo{\pi_0}
	\infer[no rule]1{\Gamma^u \seq \Delta, \phi}
	\infer[left label= \RuF]1{\Gamma\seq \Delta, \phi}
	\hypo{\pi_1}
	\infer[no rule]1{\phi^u, \Gamma^u\seq \Delta}
	\infer[left label= \RuF]1{\phi^u, \Gamma\seq \Delta}
	\infer2[\RuCut]{\Gamma \seq \Delta}
\end{prooftree}
\qquad \longrightarrow \qquad
\begin{prooftree}
	\hypo{\pi_0}
	\infer[no rule]1{\Gamma^u \seq \Delta, \phi}
	\hypo{\pi_1}
	\infer[no rule]1{\phi^u, \Gamma^u\seq \Delta}
	\infer[left label= \RuCut]2{\Gamma^u \seq \Delta}
	\infer[left label= \RuF]1{\Gamma\seq \Delta}
\end{prooftree}
\]
If the rule at the root of $\pi_1$ is different to $\RuF$ we do the following:	

\[
	\begin{prooftree}
		\hypo{\pi_0}
		\infer[no rule]1{\Gamma' \seq \Delta, \phi}
		\infer[left label= \RuDischarge]1{\Gamma' \seq \Delta, \phi}
		\infer[left label= \RuU]1{\Gamma^u\seq \Delta, \phi}
		\infer[left label=\RuF]1{\Gamma\seq \Delta, \phi}
		\hypo{\pi_1}
		\infer[no rule]1{\phi^u, \Gamma\seq \Delta}
		\infer[left label= \RuCut]2{\Gamma \seq \Delta}
	\end{prooftree}
	\qquad \longrightarrow \qquad
	\begin{prooftree}
		\hypo{\pi_0'}
		\infer[no rule]1{\Gamma^u\seq \Delta,\phi}
		\infer[left label= \RuF]1{\Gamma\seq \Delta, \phi}
		\hypo{\pi_1}
		\infer[no rule]1{\phi^u, \Gamma\seq \Delta}
		\infer[left label= \RuCut]2{\Gamma \seq \Delta}
	\end{prooftree}
	\]
	where $\pi'_0$ is defined as above. For a rule \Ru different from \RuU we proceed as follows.	
	\begin{multline*}
		\begin{prooftree}
			\hypo{\pi_1}
			\infer[no rule]1{\Gamma_1 \seq \Delta_1, \phi}
			\hypo{\cdots}
			\hypo{\pi_n}
			\infer[no rule]1{\Gamma_n \seq \Delta_n, \phi}
			\infer[left label= \Ru]3{\Gamma \seq \Delta, \phi}
			\infer[left label= \RuF]1{\Gamma\seq \Delta, \phi}
			\hypo{\pi_0}
			\infer[no rule]1{\phi^u, \Gamma\seq \Delta}
			\infer[left label= \RuCut]2{\Gamma \seq \Delta}
		\end{prooftree} \\
		 \longrightarrow \qquad
		\begin{prooftree}
			\hypo{\pi_1}
			\infer[no rule]1{\Gamma_1 \seq \Delta_1, \phi}
			\infer[left label= \RuF]1{\Gamma_1'\seq \Delta_1, \phi}
			\hypo{\cdots}
			\hypo{\pi_n}
			\infer[no rule]1{\Gamma_n \seq \Delta_n, \phi}
			\infer[left label= \RuF]1{\Gamma_n' \seq \Delta_n, \phi}
			\infer[left label= \Ru]3{\Gamma \seq \Delta, \phi}
			\hypo{\pi_0}
			\infer[no rule]1{\phi^u, \Gamma\seq \Delta}
			\infer[left label= \RuCut]2{\Gamma \seq \Delta}
		\end{prooftree}
	\end{multline*}
\subsection{Non-principal cut-reductions for a rule $\Ru$ different from \RuDia, \RuU, \RuF and \RuDischarge[]}
	\begin{multline*}
		\begin{prooftree}
			\hypo{\pi_0}
			\infer[no rule]1{\Gamma\seq \Delta, \phi}
			\hypo{\pi_1}
			\infer[no rule]1{\phi^u,\Gamma_1 \seq \Delta_1 }
			\hypo{\cdots}
			\hypo{\pi_n}
			\infer[no rule]1{\phi^u,\Gamma_n \seq \Delta_n}
			\infer[left label= \Ru]3{\phi^u,\Gamma \seq \Delta}
			\infer[left label= \RuCut]2{\Gamma \seq \Delta}
		\end{prooftree}	\\
		\longrightarrow \qquad
		\begin{prooftree}
			\hypo{\pi_0}
			\infer[no rule]1{\Gamma\seq \Delta, \phi}
			\hypo{\pi_1}
			\infer[no rule]1{\phi^u,\Gamma_1 \seq \Delta_1 }
			\infer[left label= \RuCut]2{\Gamma_1, \Gamma \seq \Delta_1, \Delta}
			\hypo{\cdots}
			\hypo{\pi_0}
			\infer[no rule]1{\Gamma\seq \Delta, \phi}
			\hypo{\pi_n}
			\infer[no rule]1{\phi^u,\Gamma_n \seq \Delta_n}
			\infer[left label= \RuCut]2{\Gamma_n, \Gamma \seq \Delta_n, \Delta}
			\infer[left label= \Ru]3{\Gamma \seq \Delta}
		\end{prooftree}
	\end{multline*}
	
	\begin{multline*}
		\begin{prooftree}
			\hypo{\pi_1}
			\infer[no rule]1{\Gamma_1 \seq \Delta_1, \phi}
			\hypo{\cdots}
			\hypo{\pi_n}
			\infer[no rule]1{\Gamma_n \seq \Delta_n, \phi}
			\infer[left label= \Ru]3{\Gamma \seq \Delta, \phi}
			\hypo{\pi_0}
			\infer[no rule]1{\phi^u, \Gamma\seq \Delta}
			\infer[left label= \RuCut]2{\Gamma \seq \Delta}
		\end{prooftree}	\\
		\longrightarrow \qquad
		\begin{prooftree}
			\hypo{\pi_1}
			\infer[no rule]1{\Gamma_1 \seq \Delta_1, \phi}
			\hypo{\pi_0}
			\infer[no rule]1{\phi^u, \Gamma\seq \Delta}
			\infer[left label= \RuCut]2{\Gamma_1, \Gamma \seq \Delta_1, \Delta}
			\hypo{\cdots}
			\hypo{\pi_n}
			\infer[no rule]1{\Gamma_n \seq \Delta_n, \phi}
			\hypo{\pi_0}
			\infer[no rule]1{\phi^u, \Gamma\seq \Delta}
			\infer[left label= \RuCut]2{\Gamma_n, \Gamma \seq \Delta_n, \Delta}
			\infer[left label= \Ru]3{\Gamma \seq \Delta}
		\end{prooftree}
	\end{multline*}
	Note, the rule \Ru in this case may also be an instance of \RuCut.

\section{Appendix: Important Cuts}\label{sec.appendixImportantCuts}
The \emph{size} of a \GKe proof $\pi$ is the number of nodes in $\pi$.
For a path $\tau$ in $\pi$ we call a path $\tau'$ a \emph{subpath} of $\tau$, if $\tau$ can be written as $\tau = \sigma_0 \tau \sigma_1$ for some paths $\sigma_0,\sigma_1$ in $\pi$.

\begin{lemma}\label{lem.pathInCluster}
	Let $\tau$ be a path through a \GKe proof $\pi$ of size $n$. If $l(\tau) \geq k \cdot n$, then there is a cluster $S$ of $\pi$ of size $n_S$ such that there is a subpath $\tau'$ of $\tau$ in $S$ with $l(\tau') \geq k \cdot n_S$.
\end{lemma}
\begin{proof}
	Let $S_1,...,S_m$ be the clusters of $\pi$. Then $|\pi| = \sum_{j=1}^{m} |S_j|$. The lemma follows, as for $i\neq j$ there is either no path from $S_i$ to $S_j$ or no path from $S_j$ to $S_i$.
\end{proof}

\begin{lemma}\label{lem.pathSuccesful}
	Let $S$ be a non-trivial cluster of size $n$ in a \GKe proof $\pi$ and $\tau = v_1v_2...$ be a path $\tau$ in $S$. Let $l(\tau)$ be the \emph{length} of the path $\tau$. Then,
	\begin{enumerate}
		\item if $l(\tau) \geq n$, then $\tau$ is successful, \label{lem.pathSucc.Succ}
		\item if $l(\tau) \geq n$, then there is a companion node on $\tau$, \label{lem.pathSucc.Comp}
		\item if $l(\tau) \geq n$, then there is a node labelled by \RuDia on $\tau$, \label{lem.pathSucc.Dia}
		\item if $l(\tau) \geq 2\cdot k\cdot n$ then there are $k$ indices $1 \leq m_1 < \cdots < m_k \leq 2\cdot k\cdot n$ such that $v_{m_1},...,v_{m_k}$ are companion nodes and the paths $v_{m_i}v_{m_i+1}...v_{m_j}$ are successful for all $i<j$. \label{lem.pathSucc.Kmany}
	\end{enumerate} 
\end{lemma}
\begin{proof}
	Every node of $\tau$ is in $S$, hence it has a formula in focus. Thus for \ref{lem.pathSucc.Succ} we only have to show that $\tau$ passes through an application of \RuFL, where the principal formula is in focus.
	If every node of $S$ is visited, then \ref{lem.pathSucc.Succ}, \ref{lem.pathSucc.Comp} and \ref{lem.pathSucc.Dia} are satisfied. Otherwise, there has to be a node that is visited twice. This is only possible if there is a discharged leaf $u$ such that every node of $\tau(u)$ is in $\tau$. Thus \ref{lem.pathSucc.Succ}, \ref{lem.pathSucc.Comp} and \ref{lem.pathSucc.Dia} are satisfied. \ref{lem.pathSucc.Kmany} follows by combining \ref{lem.pathSucc.Comp} and \ref{lem.pathSucc.Succ}.
\end{proof}

We now give a detailed proof of Lemma \ref{lem.cutsImportant}.
\setcounter{theorem}{\getrefnumber{lem.cutsImportant}}
\addtocounter{theorem}{-1}
\renewcommand\thetheorem{\arabic{theorem}}
\begin{lemma}
	Let $\pi$ be a \GKe proof with a critical cut of rank $n+1$ at the root, i.e. $\pi$ is of the form
	\begin{align*}
		\begin{prooftree}
			\hypo{\pi_0}
			\infer[no rule]1{\Sigma_0 \seq \Pi_0, \lF \phi}
			\hypo{\pi_1}
			\infer[no rule]1{\lF \phi^u, \Sigma_1 \seq \Pi_1}
			\infer[left label= \RuCut]2{\Sigma_0, \Sigma_1 \seq \Pi_0, \Pi_1}
		\end{prooftree}
	\end{align*}	
	where $\pi_0$ and $\pi_1$ are of cut-rank $n$, and the cut formula $\lF\phi$ has rank $n+1$. Then we can construct a proof $\pi'$ of $\Sigma_0,\Sigma _1 \seq \Pi_0,\Pi_1$ of cut-rank $n$.
\end{lemma}
\begin{proof}
	Let $\rho_I$ be the $\lF\phi$-traversed proof of $\Sigma_0,\Sigma _1 \seq \Pi_0,\Pi_1$ consisting of an open leaf labelled by $\Sigma_0,\Sigma _1 \seq \Pi_0,\Pi_1$ together with $\merge{{\pi}_0}{\psi}{{\pi}_1}$, this we will denote by
	\begin{align*}
		\begin{prooftree}
			\hypo{\merge{{\pi}_0}{\psi}{{\pi}_1}}
			\infer[no rule]1{\Sigma_0,\Sigma_1 \seq \Pi_0,\Pi_1}
		\end{prooftree}
	\end{align*}
	Starting from $\rho_I$ we will inductively transform the traversed proof, until we end up with a traversed proof of $\Sigma_0,\Sigma_1 \seq \Pi_0,\Pi_1$, where all leaves are closed. This will be done as follows: Let $\rho$ be a traversed proof. If all leaves are closed we are done. Otherwise consider the leftmost open leaf $v$ labelled by 
	\begin{align*}
		\begin{prooftree}
			\hypo{\merge{\pi_l}{\psi}{\pi_r}}
			\infer[no rule]1{\Gamma_l,\Gamma_r \seq \Delta_l,\Delta_r}
		\end{prooftree}
	\end{align*}
	Let the roots of $\pi_l$ and $\pi_r$ be labelled by the rules $\Ru_l$ and $\Ru_r$, respectively. We transform $\rho$ by a case distinction on $\Ru_l$ and $\Ru_r$.
	
	\begin{itemize}	
		\item If $\Ru_l$ is \RuDischarge, then $\pi_l$ has the form
		\begin{align*}
			\begin{prooftree}
				\hypo{\pi_l'}
				\infer[no rule]1{\Gamma_l \seq \Delta_l, \psi}
				\infer1[\RuDischarge]{\Gamma_l \seq \Delta_l,\psi}
			\end{prooftree}
		\end{align*}
				We proceed by making the following case distinction.
		\begin{itemize}
			\item If there is a node $c$ in $\rho$, that is an ancestor of $v$ and labelled by $\Gamma_l,\Gamma_r \seq \Delta_l,\Delta_r$, such that the path from $c$ to $v$ is successful, then insert a \RuDischarge[\dz] rule at $c$ and let $v$ be the discharged leaf $[\Gamma_l,\Gamma_r \seq \Delta_l,\Delta_r]^\dz$ with fresh discharge token $\dz$.			
			\item Else, unfold $\pi_l$, i.e. let $\tilde{\pi}_l$ be the proof obtained from $\pi_l'$ by replacing every discharged leaf labelled by $\dx$ with $\pi_l$.
			We replace $v$ by
			\begin{align*}
				\begin{prooftree}
					\hypo{\merge{\tilde{\pi}_l}{\psi}{\pi_r}}
					\infer[no rule]1{\Gamma_l,\Gamma_r \seq \Delta_l,\Delta_r}
				\end{prooftree}
			\end{align*}
		\end{itemize}

		\item If $\Ru_l$ is \RuFR with principal formula $\lF \phi$ we make a case distinction on $\Ru_r$:
		\begin{itemize}
			\item If $\Ru_r$ is \RuFL with principal formula $\lF \phi$, the proofs $\pi_l$ and $\pi_r$ have the form 
			\begin{align*}
					\begin{prooftree}
						\hypo{\pi_l'}
						\infer[no rule]1{\Gamma_l \seq \Delta_l, \phi, \ldia\lF\phi}
						\infer[left label= \RuFR]1{\Gamma_l \seq \Delta_l,\lF\phi}
					\end{prooftree}
		&\qquad
					\begin{prooftree}
						\hypo{\pi_r'}
						\infer[no rule]1{\ldia \lF \phi^a, \Gamma_r \seq \Delta_r}
						\hypo{\pi_r^u}
						\infer[no rule]1{\phi^u, \Gamma_r \seq \Delta_r}
						\infer[left label= \RuFL]2{\lF\phi^a,\Gamma_r \seq  \Delta_r}
					\end{prooftree}
			\end{align*}
			and $v$ is replaced by
			\begin{align*}
				\begin{prooftree}
					\hypo{\merge{\pi_l'}{\ldia\lF\phi}{\pi_r'}}
					\infer[no rule]1{\Gamma_l,\Gamma_r \seq \Delta_l,\Delta_r,\phi}
					\hypo{\pi_r^u}
					\infer[no rule]1{\phi^u, \Gamma_r \seq \Delta_r}
					\infer[left label= \RuCut]2{\Gamma_l,\Gamma_r \seq \Delta_l,\Delta_r}
				\end{prooftree}
			\end{align*}
			where the introduced cut has rank $\rank(\phi)$.
			
			\item If $\Ru_r$ is $\RuU$ of the form
			\begin{align*}
				\begin{prooftree}
					\hypo{\pi_r'}
					\infer[no rule]1{\psi^a, \Gamma_r' \seq \Delta_r}
					\infer[left label= \RuU]1{\psi^u, \Gamma_r \seq \Delta_r}
				\end{prooftree}
				&
	\qquad\qquad\text{then $v$ is replaced by} 
		&
				\begin{prooftree}
					\hypo{\merge{\pi_l}{\psi}{\pi_r'}}
					\infer[no rule]1{\Gamma_l,\Gamma_r \seq \Delta_l,\Delta_r}
				\end{prooftree}
			\end{align*}
			
			\item If $\Ru_r$ is $\RuF$ of the form
			\begin{align*}
				\begin{prooftree}
					\hypo{\pi_r'}
					\infer[no rule]1{\psi^u, \Gamma_r^u \seq \Delta_r}
					\infer[left label= \RuF]1{\psi^a, \Gamma_r \seq \Delta_r}
				\end{prooftree}
			&
	\qquad\qquad\text{then $v$ is replaced by} 
		&
				\begin{prooftree}
					\hypo{\merge{\pi_l}{\psi}{\pi_r'}}
					\infer[no rule]1{\Gamma_l,\Gamma_r \seq \Delta_l,\Delta_r}
				\end{prooftree}
			\end{align*}
			
			\item If $\Ru_r$ is \RuDischarge[\dy] of the form
			\begin{align*}
				\begin{prooftree}
					\hypo{\pi_r'}
					\infer[no rule]1{\psi^a,\Gamma_r \seq \Delta_r}
					\infer[left label= \RuDischarge[\dy]]1{\psi^a, \Gamma_r \seq \Delta_r}
				\end{prooftree}
			\end{align*}
			then unfold $\pi_r$, i.e. let $\tilde{\pi}_r$ be the proof obtained from $\pi_r'$ by replacing every discharged leaf labelled by $\dy$ with $\pi_r$.\footnote{Discharge tokens $\dz$ are replaced by fresh discharge tokens, whenever a \RuDischarge[\dz] rule is duplicated.}
			We replace $v$ by
			\begin{align*}
				\begin{prooftree}
					\hypo{\merge{\pi_l}{\psi}{\tilde{\pi}_r}}
					\infer[no rule]1{\Gamma_l,\Gamma_r \seq \Delta_l,\Delta_r}
				\end{prooftree}
			\end{align*}

			\item If $\Ru_r$ is \RuWeak of the form
			\begin{align*}
				\begin{prooftree}
					\hypo{\pi_r'}
					\infer[no rule]1{\Gamma_r' \seq \Delta_r}
					\infer[left label= \RuWeak]1{\psi^u, \Gamma_r' \seq \Delta_r}
				\end{prooftree}
			&
	\qquad\qquad\text{then $v$ is replaced by} 
		&
				\begin{prooftree}
					\hypo{\pi_r'}
					\infer[no rule]1{\Gamma_r' \seq \Delta_r}
					\infer[left label= \RuU]1{\Gamma_r \seq \Delta_r}
					\infer[left label= \RuWeak]1{\Gamma_l, \Gamma_r \seq \Delta_l,\Delta_r}
				\end{prooftree}
			\end{align*}

			\item If $\Ru_r$ is any other rule, then the proof has the form 
			\begin{align*}
				\begin{prooftree}
					\hypo{\pi_r^1}
					\infer[no rule]1{\psi^a,\Gamma_r^1 \seq \Delta_r^1}
					\hypo{\hdots}				
					\hypo{\pi_r^n}
					\infer[no rule]1{\psi^a,\Gamma_r^n \seq \Delta_r^n}
					\infer[left label= \Ru]3{\psi^a,\Gamma_r \seq \Delta_r}
				\end{prooftree}
			\end{align*}
			and $v$ is replaced by
			\begin{align*}
				\begin{prooftree}
					\hypo{\merge{\pi_l}{\psi}{\pi_r^1}}
					\infer[no rule]1{\Gamma_l, \Gamma_r^1 \seq \Delta_l,\Delta_r^1}
					\hypo{\hdots}				
					\hypo{\merge{\pi_l}{\psi}{\pi_r^n}}
					\infer[no rule]1{\Gamma_l, \Gamma_r^n \seq \Delta_l,\Delta_r^n}
					\infer[left label= \Ru]3{\Gamma_l, \Gamma_r \seq \Delta_l,\Delta_r}
				\end{prooftree}
			\end{align*}
		\end{itemize}

		\item If $\Ru_l$ is \RuDia, then we make a case distinction on $\Ru_r$:
		\begin{itemize}
			\item If $\Ru_r$ is \RuDia, the proofs have the form 
			\begin{align*}
				\begin{minipage}{0.35\textwidth}
					\begin{prooftree}
						\hypo{\pi_l'}
						\infer[no rule]1{\chi^a \seq \Delta_l, \lF \phi}
						\infer[left label= \RuDia]1{\ldia\chi^a \seq \ldia\Delta_l, \ldia\lF\phi}
					\end{prooftree}
				\end{minipage}
				\begin{minipage}{0.35\textwidth}
					\begin{prooftree}
						\hypo{\pi_r'}
						\infer[no rule]1{\lF\phi^b \seq \Delta_r}
						\infer[left label= \RuDia]1{\ldia\lF\phi^b \seq \ldia\Delta_r}
					\end{prooftree}
				\end{minipage}
			\end{align*}
	and $v$ is replaced by
			\begin{align*}
				\begin{prooftree}
					\hypo{\merge{\pi_l'}{\lF\phi}{\pi_r'}}
					\infer[no rule]1{\chi^a \seq \Delta_l, \Delta_r}
					\infer[left label= \RuDia]1{\ldia\chi^a \seq \ldia\Delta_l,\ldia\Delta_r}
				\end{prooftree}
			\end{align*}
			
			\item If $\Ru_r$ is any other rule, we do the same as in the previous case.
		\end{itemize}
		
			\item If $\Ru_l$ is \RuU of the form 
		\begin{align*}
			\begin{prooftree}
				\hypo{\pi_l'}
				\infer[no rule]1{\Gamma_l' \seq \Delta_l, \psi}
				\infer[left label= \RuU]1{\Gamma_l \seq \Delta_l, \psi}
			\end{prooftree}
			&
	\qquad\qquad\text{then $v$ is replaced by} 
		&
			\begin{prooftree}
				\hypo{\merge{\pi_l'}{\psi}{\pi_r}}
				\infer[no rule]1{\Gamma_l',\Gamma_r \seq \Delta_l,\Delta_r}
				\infer[left label= \RuU]1{\Gamma_l,\Gamma_r \seq \Delta_l,\Delta_r}
			\end{prooftree}
		\end{align*}
		
		\item If $\Ru_l$ is \RuF of the form 
		\begin{align*}
			\begin{prooftree}
				\hypo{\pi_l'}
				\infer[no rule]1{\Gamma_l^u \seq \Delta_l, \psi}
				\infer[left label= \RuF]1{\Gamma_l \seq \Delta_l, \psi}
			\end{prooftree}
		&
	\qquad\qquad\text{then $v$ is replaced by} 
		&
			\begin{prooftree}
				\hypo{\merge{\pi_l'}{\psi}{\pi_r}}
				\infer[no rule]1{\Gamma_l^u,\Gamma_r \seq \Delta_l,\Delta_r}
				\infer[left label= \RuF]1{\Gamma_l,\Gamma_r \seq \Delta_l,\Delta_r}
			\end{prooftree}
		\end{align*}

		\item If $\Ru_l$ is \RuWeak of the form
		\begin{align*}
			\begin{prooftree}
				\hypo{\pi_l'}
				\infer[no rule]1{\Gamma_l \seq \Delta_l}
				\infer[left label=\RuWeak]1{\Gamma_l \seq \Delta_l,\psi}
			\end{prooftree}
		&
	\qquad\qquad\text{then $v$ is replaced by} 
		&
			\begin{prooftree}
				\hypo{\pi_l'}
				\infer[no rule]1{\Gamma_l \seq \Delta_l}
				\infer[left label=\RuWeak]1{\Gamma_l, \Gamma_r \seq \Delta_l,\Delta_r}
			\end{prooftree}
		\end{align*}
		
		\item If $\Ru_l$ is any other rule, $\pi_l$ has the form
		\begin{align*}
			\begin{prooftree}
				\hypo{\pi_l^1}
				\infer[no rule]1{\Gamma_l^1 \seq \Delta_l^1, \psi}
				\hypo{\hdots}				
				\hypo{\pi_l^n}
				\infer[no rule]1{\Gamma_l^n \seq \Delta_l^n, \psi}
				\infer[left label= \Ru]3{\Gamma_l \seq \Delta_l,\psi}
			\end{prooftree}
		\end{align*}
and $v$ is replaced by
		\begin{align*}
			\begin{prooftree}
				\hypo{\merge{\pi_l^1}{\psi}{\pi_r}}
				\infer[no rule]1{\Gamma_l^1, \Gamma_r \seq \Delta_l^1,\Delta_r}
				\hypo{\hdots}				
				\hypo{\merge{\pi_l^n}{\psi}{\pi_r}}
				\infer[no rule]1{\Gamma_l^n, \Gamma_r \seq \Delta_l^n,\Delta_r}
				\infer[left label= \Ru]3{\Gamma_l, \Gamma_r \seq \Delta_l,\Delta_r}
			\end{prooftree}
		\end{align*}
		
	\end{itemize}
	
	Let $\rho_i$ and $\rho_j$ be traversed proofs. We write $\rho_i < \rho_j$ if $\rho_j$  can be obtained from $\rho_i$ by the above construction and $\rho_i \neq \rho_j$. It holds that $<$ is irreflexive, antisymmetric and transitive.
	Moreover, if $\rho_i < \rho_j$, then $\rho_i$ is a subproof of $\rho_j$, in the sense that $\rho_j$ can be obtained from $\rho_i$ by replacing some open leaves in $\rho_i$ by traversed proofs and inserting nodes labelled by \RuDischarge[]. Thus, $\rho_j$ consists of at least the nodes in $\rho_i$ and we can identify nodes in $\rho_i$ with nodes in $\rho_j$.
	
	From now on, whenever we speak about a traversed proof $\rho$ we mean a traversed proof $\rho > \rho_I$.
	
	\bigskip
	\noindent
	Let $\rho > \rho_I$ and let $v$ be an open leaf in $\rho$ labelled by
	\begin{align*}
		\begin{prooftree}
			\hypo{\merge{\pi_l}{\psi}{\pi_r}}
			\infer[no rule]1{\Gamma_l,\Gamma_r \seq \Delta_l,\Delta_r}
		\end{prooftree}
	\end{align*}
	We first define nodes $u_0(v) \in \pi_0$ and  $u_1(v) \in \pi_1$ with
	\begin{enumerate}
		\item $\pi_l = \pi_{u_0(v)}$ and $\pi_r = \pi_{u_1(v)}$,
		\item $\sfS(u_0(v)) = \Gamma_l \seq \Delta_l, \psi$ and $\sfS(u_1(v)) = \psi^a,\Gamma_r' \seq \Delta_r$,
	\end{enumerate}
	The nodes $u_0(v)$ and  $u_1(v)$ are defined by recursion on the construction. For $\rho_I$ define $u_0(v)$ to be the root of $\pi_0$ and $u_1(v)$ to be the root of $\pi_1$. For the recursion step we follow the case distinction. For example, let $\rho$ be a traversed proof with leftmost open leaf $v$, where the last applied rules in $\pi_l$ and $\pi_r$ are \RuDia, respectively. 
	Let $\rho'$ be obtained from $\rho$ in one step with leftmost open leaf $v'$. Then $u_0(v')$ is the child of $u_0(v)$ and $u_1(v')$ is the child of $u_1(v)$.

	
	This definition extends to nodes $w$ in a traversed proof $\rho$ below an open leaf. For such nodes $w$ and $w'$ it moreover holds, that if $w$ is the parent of $w'$, then either 
	\begin{enumerate}
		\item $u_0(w) = u_0(w')$,
		\item $u_0(w)$ is the parent of $u_0(w')$ or
		\item $u_0(w)$ is an ancestor of $u_0(w')$, where all but one node on the ancestor path are labelled by \RuDischarge[].  
	\end{enumerate}
	The same holds for the nodes $u_1(w)$ and $u_1(w')$.
	
	Let $\tau$ be the path from the root of $\rho$ to the open leaf $v$. By the above definition we can define corresponding paths $\tau_0$ in $\pi_0$ and $\tau_1$ in $\pi_1$.
	
	Let $n_0$ be the size of $\pi_0$ and $n_1$ be the size of $\pi_1$. By the above argumentation nodes in $\tau$ can only be labelled by at most $ n_0 \cdot n_1$ distinct sequents. Moreover, if $l(\tau) = k\cdot n_1$, then $l(\tau_0) \geq k$. This holds as in every step of the construction either $\pi_l$ or $\pi_r$ is transformed. Whenever $\pi_l$ is transformed case 1 above cannot be the case. But if $l(\tau) \geq n_1$, $\pi_l$ has to be transformed, as otherwise $l(\tau_1) \geq n_1$. Due to Lemma \ref{lem.pathInCluster} and \ref{lem.pathSuccesful}.\ref{lem.pathSucc.Dia} there has to be a \RuDia rule on $\tau_1$, in which case $\pi_l$ gets transformed as well.
	
	\bigskip
	We claim that every open leaf in a traversed proof $\rho$ has height at most $4 \cdot n_0^2 \cdot n_1^2$. For suppose that $v$ is an open leaf of height more than $4 \cdot n_0^2 \cdot n_1^2$. Let $\tau$ be the path from the root of $\rho$ to $v$ and let $\tau_0$ be the corresponding path in $\pi_0$. It holds that $l(\tau_0)\geq 4 \cdot n_0^2 \cdot n_1$. Lemma \ref{lem.pathInCluster} gives a cluster $S$ with size $n_S$ of $\pi_0$ and a subpath $\tau_S$ in $S$ of $\tau_0$ such that $l(\tau_S) \geq 4 \cdot n_S \cdot n_0 \cdot n_1$. Using Lemma \ref{lem.pathSuccesful}.\ref{lem.pathSucc.Kmany} we obtain $2 \cdot n_0 \cdot n_1$ companion nodes $c_1,c_2,...$ in $S$. Hence there are as many traversed proofs with leftmost open leaves $v_1,v_2,...$ with $u_0(v_j) = c_j$ for all $j$. This means that there are $2 \cdot n_0 \cdot n_1$ traversed proofs $\rho_1 < \rho_2 < \cdots < \rho$, with left-most open leaves $v_1,v_2,...$, which are \emph{repeat leaves}, i.e. where $\pi_l$ is of the form
	\begin{align*}
		\begin{prooftree}
			\hypo{\pi_l'}
			\infer[no rule]1{\Gamma_l \seq \Delta_l, \psi}
			\infer[left label= \RuDischarge]1{\Gamma_l \seq \Delta_l,\psi}
		\end{prooftree}
	\end{align*} 
	By the above argumentation the nodes $v_1,v_2,...$ are labelled by at most $n_0 \cdot n_1$ many distinct sequents. Hence there are $i < j$ such that $\sfS(v_i) = \sfS(v_j)$. Due to the choice of the companion nodes in Lemma 
	\ref{lem.pathSuccesful}.\ref{lem.pathSucc.Kmany} the path from $c_i$ to $c_j$ in $\tau_S$ is successful. Note that, if $u_0(w)$ is labelled by \RuFL, where the principal formula is in focus, then so is $w$. Because all sequents in $\tau_S$ have a formula in focus, so does the subpath from $v_i$ to $v_j$ of $\tau$. Hence, the path from $v_i$ to $v_j$ is successful and the leaf gets closed in the next step of transforming $\rho_j$. This contradicts the fact that $v$ in $\rho$ is an open leaf of height more than $4 \cdot n_0^2 \cdot n_1^2$.
	
	In conclusion, as every constructed tree is finitely branching, after finitely many steps a traversed proof $\rho_T$ without open leaves is constructed.	
	As every traversed proof without open leaves is a \GKe proof we have shown the lemma.
\end{proof}

\end{document}